\documentclass[11pt,reqno]{amsart}
\usepackage[T1]{fontenc}
\usepackage[utf8]{inputenc}
\usepackage{mathpazo}
\usepackage{amsmath,amssymb,amsthm,mathtools}
\usepackage[margin=1.05in]{geometry}
\usepackage{microtype}
\usepackage{xcolor}
\usepackage{graphicx}
\graphicspath{{figures/}}
\usepackage{booktabs}
\usepackage{needspace}
\usepackage{tikz}
\usetikzlibrary{arrows.meta}
\usepackage[colorlinks=true,linkcolor=blue!45!black,citecolor=blue!45!black,urlcolor=blue!45!black]{hyperref}
\hypersetup{pdftitle={Boundary parity in the mirror model on the Manhattan lattice: certified numerical experiments},pdfauthor={}}
\numberwithin{equation}{section}
\newtheorem{theorem}{Theorem}[section]
\newtheorem{proposition}[theorem]{Proposition}
\newtheorem{lemma}[theorem]{Lemma}
\newtheorem{corollary}[theorem]{Corollary}
\theoremstyle{definition}

\theoremstyle{remark}

\newcommand{\Z}{\mathbb Z}
\newcommand{\R}{\mathbb R}

\newcommand{\Pp}{\mathbb P_p}

\newcommand{\ind}{\mathbf1}
\newcommand{\ff}{\mathfrak f}

\newcommand{\cM}{\mathcal M}
\newcommand{\cR}{\mathcal R}
\newcommand{\one}{\mathbf1}
\DeclareMathOperator{\rank}{rank}
\DeclareMathOperator{\tr}{tr}

\DeclareMathOperator{\wind}{wind}
\DeclareMathOperator{\diag}{diag}

\newcommand{\FirstAuthor}{Jian Gu}
\newcommand{\FirstAffiliation}{ESSEC}
\newcommand{\SecondAuthor}{Qin Hao}
\newcommand{\SecondAffiliation}{Polytechnic Institute of Paris}
\title[Boundary parity on the Manhattan lattice]{Boundary parity in the mirror model on the Manhattan lattice}
\author[\FirstAuthor]{\FirstAuthor}
\address{\FirstAffiliation}
\author[\SecondAuthor]{\SecondAuthor}
\address{\SecondAffiliation}
\makeatletter
\renewcommand{\@setauthors}{%
  \begingroup\trivlist\centering\footnotesize
  \@topsep30\p@\relax\advance\@topsep by -\baselineskip
  \item\relax
  \begin{minipage}[t]{.46\textwidth}\centering
    {\scshape\FirstAuthor\par}\smallskip
    {\normalfont\footnotesize\FirstAffiliation\par}
  \end{minipage}\hfill
  \begin{minipage}[t]{.46\textwidth}\centering
    {\scshape\SecondAuthor\par}\smallskip
    {\normalfont\footnotesize\SecondAffiliation\par}
  \end{minipage}
  \endtrivlist\endgroup}
\renewcommand{\@setaddresses}{}
\makeatother
\date{}

\begin{document}
\raggedbottom
\begin{abstract}
We study a finite-domain event for the mirror model on the Manhattan
lattice: every boundary-to-boundary trajectory crosses a marked
connector an even number of times, while internal cycles are
unrestricted. Our main theorem gives a finite-scale confinement
criterion: if this event's probability reaches a universal
threshold in one even two-square rectangle, then almost surely every
trajectory in the plane is periodic. We construct an exactly normalized
three-color tensor representation and prove deterministic bounds for
approximate contractions using the signed partial-permutation structure
of directed continuations. These tools turn numerical evaluations into
rigorous finite-volume statements. At mirror density \(p=2/5\), we
prove that the boundary-parity probability is nonmonotone in the
rectangle width, with a decrease followed by a certified increase.
These values remain below the sufficient confinement threshold and
do not resolve planar localization at this density.
Further experiments track how the measured finite-size minimum varies
with density, compare matching-based and entrywise error bounds, and
exhibit endpoint-sensitive and cooperative responses to local mirror
changes.
\end{abstract}
\maketitle

\section{Introduction}\label{sec:intro}

The Manhattan mirror model is a deterministic routing system in an
independent random environment. Horizontal streets point east on even
rows and west on odd rows; vertical streets point north on even columns
and south on odd columns. Independently at each vertex, a mirror is
present with probability \(p\in[0,1]\). A ray continues along its
current axis at a vacant vertex and changes axis at a mirror. Repeated
passages through a physical site use the same Bernoulli variable.

Finite-domain calculations offer a way to study this dependence before
long trajectories can be controlled analytically. We investigate a
boundary-parity probability that has a rigorous connection to planar
confinement. An exact tensor representation and deterministic residual
bounds allow us to certify its finite-size behavior. The main theorem
states the confinement criterion, and the principal numerical result
proves a change of finite-size trend at \(p=2/5\).

\subsection{Main results}

Fix an even integer \(L\ge2\), put \(c=L/2\), and set
\begin{equation}\label{eq:geometry}
 D_L=\{0,\ldots,2L-1\}\times\{0,\ldots,L-1\}.
\end{equation}
The horizontal connector \(\alpha_L\) joins the centers
\begin{equation}\label{eq:connector}
 u_L=(c-\tfrac12,c-\tfrac12),\qquad
 v_L=(3c-\tfrac12,c-\tfrac12)
\end{equation}
of the two square halves. Boundary paths are the complete directed
routes in \(D_L\) from an incoming boundary leg to an outgoing
boundary leg; the remaining routes are internal cycles. Define
\begin{equation}\label{eq:observable}
 H_L=\{\text{every boundary path crosses }\alpha_L
                         \text{ an even number of times}\},\qquad
 h_L(p)=\mathbb P_p(H_L).
\end{equation}
Intersections are counted with traversal multiplicity, and the parity
condition is imposed separately on every boundary path. Internal cycles
are unrestricted. Section~\ref{sec:model} gives the directed-state
formulation and Figure~\ref{fig:connector} illustrates the geometry.

\Needspace{7\baselineskip}
\begin{theorem}[Finite-scale parity criterion]\label{thm:criterion}
If, for some finite even \(L\),
\begin{equation}\label{eq:criterion}
 h_L(p)\ge \frac{8457}{10000},
\end{equation}
then almost surely every trajectory of the Manhattan mirror model
at parameter \(p\) is periodic.
\end{theorem}

The threshold in \eqref{eq:criterion} comes from the surrounding-circuit
theorem of Balister, Johnston, Savery, and Scott~\cite{BalisterEtAl}.
Our contribution is the boundary event and its deterministic blocking
property: boundary paths may cross the connector any even number of
times, while internal cycles impose no condition. The proof of
Theorem~\ref{thm:criterion} is given in Appendix~\ref{app:criterion}.

The same event admits exact computation and rigorous approximation.
The three-color tensor in Theorem~\ref{thm:tensor} represents its
probability with total weight one for every internal cycle. At a
directed cut, each fixed-environment continuation is a signed partial
permutation (Theorem~\ref{thm:continuation}). Nonnegative rational row
and column potentials covering the exact residuals therefore certify
an approximate contraction (Theorem~\ref{thm:certificate}). Applying
these results yields the following finite-volume theorem.

\Needspace{8\baselineskip}
\begin{theorem}[Certified finite-size nonmonotonicity]\label{thm:nonmonotonicity}
At mirror density \(p=2/5\),
\[
 h_L(2/5)>h_{L+2}(2/5)
 \qquad\text{for }L\in\{2,4,6,8,10,12\},
\]
whereas
\[
 h_{16}(2/5)-h_{14}(2/5)>\frac{34067}{10^8}.
\]
Thus \(L\mapsto h_L(2/5)\) is nonmonotone on the even widths,
and \(L=14\) is the unique minimizer among the tested widths
\(2,4,\ldots,16\).
\end{theorem}

The proof uses exact integer evaluations and rational enclosures,
reported in Section~\ref{sec:finite-size}. The strict increase rules
out a decreasing extrapolation of the smaller-width values. It does
not determine the large-width limit. The computed probabilities remain
below the threshold in Theorem~\ref{thm:criterion}, so this theorem is
not used to claim planar localization at \(p=2/5\).

\subsection{Further experiments and organization}

We complement the finite-size theorem with an exact density scan. The
location of the smallest measured value changes with density, leading
to crossings of finite-size curves. We report these comparisons
without fitting a critical density or a scaling exponent.

A compression experiment compares matrix-product-state proposals
against the exact transfer maps. For the same proposals, matching-based
residual bounds are consistently smaller than entrywise bounds.
However, an accurate central value can still have a substantially
larger certified error. Both quantities are measured separately.

Finally, complete enumerations in two \(4\times4\) boxes examine
the response to local mirror changes. They exhibit dependence on the
connector endpoints and a cooperative two-contact response invisible
to either isolated change. These tests identify constraints on local
simplifications of the boundary observable.

Sections~\ref{sec:finite-size}--\ref{sec:local-experiments} present
the finite-size, density, compression, and local-routing experiments.
Section~\ref{sec:methods} explains their computational procedures and
verification. The supporting proofs of the confinement criterion,
tensor representation, error bounds, and local-response formulas are
collected in the appendices.

\subsection{Relation to earlier work}

Lattice versions of the wind-tree model were introduced by Ruijgrok and
Cohen~\cite{RuijgrokCohen}. Numerical studies of related Lorentz lattice
gases have examined both transport and the underlying trajectory geometry.
Binder~\cite{Binder1987} measured diffusion and finite-state cycle
statistics for several lattice rules. Kong and
Cohen~\cite{KongCohen1989,KongCohen1991} investigated how transport
depends on the scattering rule and lattice geometry. Further numerical
comparisons by Cohen and Wang~\cite{CohenWang1995} and Wang and
Cohen~\cite{WangCohen1995} explored several Lorentz lattice-gas
models, including honeycomb and quasi-lattice geometries.

For random-mirror trajectories, Ziff, Kong, and
Cohen~\cite{ZiffKongCohen} developed connections with kinetic walks
and percolation geometry. Owczarek and
Prellberg~\cite{OwczarekPrellberg} used kinetic-growth simulations
of interacting trails, while Cao and Cohen~\cite{CaoCohen}
numerically studied the scaling of closed trajectories in mirror and
rotator models. These studies provide context for the sensitivity of
finite-size measurements to the local routing law. Classical trajectory
representations of quantum networks were developed by Beamond, Cardy,
and Chalker~\cite{BeamondCardyChalker} and, on the Manhattan lattice,
by Beamond, Owczarek, and Cardy~\cite{BeamondOwczarekCardy}.

Recent work also addresses finite-slab transport. Lefevere~\cite{LefevereConductivity}
compares three-dimensional Lorentz conductivity simulations with a
multiscale calculation based on a closure assumption. Lefevere and
Tasaki~\cite{LefevereTasaki} combine a hierarchical mirror model with
simulations of conductance fluctuations in the original three-dimensional
model. Gu and Hao~\cite{GuHao} study boundary transmission and
conditional excursion geometry in the symmetric planar Lorentz mirror
model, reporting cylinder penetration, repeated visits, and
endpoint-collision statistics with sampling uncertainty. The present
experiments concern the Manhattan mirror law and a per-boundary-path
parity event, with exact summation and deterministic contraction bounds
controlling numerical error.

For the planar Manhattan mirror model, Li~\cite{LiManhattan} proved exponential
confinement for \(p>1/2-\varepsilon_0\), for some
\(\varepsilon_0>0\). Polynomial confinement on even cylinders was
proved by Li~\cite{LiCylinder}; Ryan~\cite{Ryan} obtained a
\(p^{-2}\) scale in a low-density cylinder regime. The Lorentz
escape lower bound of Kozma and Sidoravicius~\cite{KS} makes it
necessary to distinguish the two models when interpreting finite-size
data. The computations reported here do not establish planar
localization at \(p=2/5\).

Our contraction uses a three-color realization of the signed loop
cancellation appearing in the intersecting-loop and supersymmetric
literature~\cite{MartinsNienhuisRietman,JacobsenReadSaleur,NahumEtAl}.
Matrix-product-state and tensor-train approximation are standard
tools~\cite{Oseledets}; certified contraction through convex
relaxations has also been studied by Ono, Zhang, and
Po~\cite{OnoZhangPo}. The residual covers used below are based on
classical matching inequalities~\cite{Birkhoff,Kuhn}. Their applicability
here follows from the directed boundary connections of the actual
physical model, rather than from an assumption about a numerical
optimizer.
Our study is motivated by the broader problem of localization in random
media, originating in Anderson's work \cite{Anderson} and the subsequent
scaling theory of Abrahams, Anderson, Licciardello, and Ramakrishnan
\cite{AALR}. For random Schr\"odinger operators, rigorous results on
spectral localization include
\cite{DingSmart,LiBernoulli2D,LiZhangBernoulli3D}. Here we investigate
geometric confinement through estimates on trajectories generated by
random local pairings.

\section{Observable and experimental design}\label{sec:model}

\subsection{One variable per physical site}

Rows with even coordinate point east and odd rows point west; columns
with even coordinate point north and odd columns point south. Put
\[
 H_y=((-1)^y,0),\qquad V_x=(0,(-1)^x).
\]
The independent variables \(\omega_{x,y}\) satisfy
\(\mathbb P_p(\omega_{x,y}=1)=p\), where one denotes a mirror.
An incoming state \((x,y,A)\), with \(A\in\{H,V\}\), continues
on its current axis if \(\omega_{x,y}=0\) and changes axis if
\(\omega_{x,y}=1\). It then moves by \(H_y\) or \(V_x\),
according to its outgoing axis. Both passages through a site use
this same bit. Unless stated otherwise, \(p\in[0,1]\).

The routing is a bijection on directed states. In a finite rectangular
domain, cut edges at the geometric boundary of the union of unit
squares centered at its vertices. Each incoming boundary leg then has
a unique outgoing partner. The remaining states form internal cycles;
a boundary path cannot enter one of those cycles. A proof is given in
Lemma~\ref{lem:decomposition}.

For the domain \(D_L\) in \eqref{eq:geometry}, the connector
\(\alpha_L\) with endpoints \eqref{eq:connector} crosses the
\(L\) vertical edges between rows \(c-1\) and \(c\), at columns
\(c\le x<3c\), where \(c=L/2\). The boundary decomposition
therefore makes \(H_L\) in \eqref{eq:observable} a well-defined
event of the physical mirror field. Figure~\ref{fig:connector} fixes
the geometry used in all balanced-rectangle experiments.

\begin{figure}[htbp]
\centering
\begin{tikzpicture}[scale=.72,>=Stealth]
\draw[gray!65,dashed](-.5,-.5) rectangle (7.5,3.5);
\draw[gray!65,dashed](3.5,-.5)--(3.5,3.5);
\foreach \y in {0,1,2,3}{
 \foreach \x in {0,1,2,3,4,5,6}{
  \pgfmathtruncatemacro{\evenrow}{mod(\y,2)}
  \ifnum\evenrow=0
    \draw[->,gray!75](\x,\y)--(\x+1,\y);
  \else
    \draw[<-,gray!75](\x,\y)--(\x+1,\y);
  \fi}}
\foreach \x in {0,1,2,3,4,5,6,7}{
 \foreach \y in {0,1,2}{
  \pgfmathtruncatemacro{\evencol}{mod(\x,2)}
  \ifnum\evencol=0
    \draw[->,gray!75](\x,\y)--(\x,\y+1);
  \else
    \draw[<-,gray!75](\x,\y)--(\x,\y+1);
  \fi}}
\foreach \x in {2,3,4,5}{\draw[blue!65!black,very thick](\x,1.2)--(\x,1.8);}
\foreach \x in {0,1,2,3,4,5,6,7}
 \foreach \y in {0,1,2,3}{\fill (\x,\y) circle(1.6pt);}
\draw[red!70!black,thick,->](1.5,1.5)--(5.5,1.5);
\fill[red!70!black](1.5,1.5) circle(2.3pt);
\fill[red!70!black](5.5,1.5) circle(2.3pt);
\node[above,red!70!black] at (1.5,1.5){\(u_L\)};
\node[above,red!70!black] at (5.5,1.5){\(v_L\)};
\node[below,red!70!black] at (3.5,1.5){\(\alpha_L\)};
\end{tikzpicture}
\caption{The two-square rectangle for \(L=4\). The connector crosses four
marked vertical edges. The arrows specify the deterministic street
directions; mirror variables are assigned independently to the vertices.}
\label{fig:connector}
\end{figure}
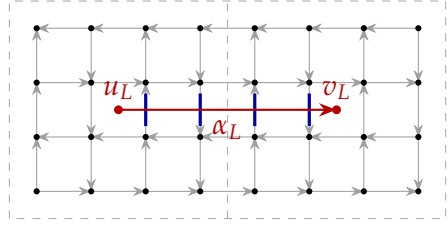

For the smaller local tests we also use \(H(D,\alpha)\), the same
per-path event for a specified rectangle and connector. An edge is
marked if its intersection multiplicity with \(\alpha\) is odd.
With sign \(-1\) on marked edges and \(+1\) on the others, a
boundary path is successful exactly when its sign is positive.
Thus \(H(D,\alpha)\) is equivalent to \(S\one=\one\), where
\(S\) is the signed boundary permutation.

\subsection{The four experiments}

Table~\ref{tab:design} records the domains, parameters, and arithmetic
used below. At rational \(p=a/b\), the full physical product measure
can be contracted with integers:
\begin{equation}\label{eq:intro-integer}
 h_L(a/b)=\frac{Z_L(a,b)}{b^{2L^2}},\qquad Z_L(a,b)\in\Z.
\end{equation}
The normalized three-color representation establishing this identity
is described in Section~\ref{sec:methods} and proved in
Appendix~\ref{sec:tensor}.

\begin{table}[htbp]
\centering\small
\begin{tabular}{p{.20\textwidth}p{.39\textwidth}p{.33\textwidth}}
\toprule
Experiment&Domain and parameters&Reported output\\
\midrule
Finite size&\(D_L\), \(p=2/5\), even \(L=2,\ldots,16\)
&Exact values through \(L=14\); certified interval at \(L=16\)\\[3pt]
Density scan&\(L=2,4,6,8,10\); \(p=k/20\), \(0\le k\le20\)
&105 exact rational probabilities\\[3pt]
Compression&\(L=10\), \(p=2/5\), eight bond caps; coefficient denominator \(2^{40}\)
&Actual errors and two deterministic error bounds\\[3pt]
Local routing&Two marked \(4\times4\) boxes; all fields and all single-site changes
&Cavity counts, signed responses, and contact examples\\
\bottomrule
\end{tabular}
\caption{Experimental design. The \(4\times4\) local boxes are
separate test geometries and are not the \(D_4\) rectangle.}
\label{tab:design}
\end{table}

The probabilities are obtained by exact summation or by deterministic
enclosure of a tensor contraction. No sampling uncertainty enters
these measurements. Floating-point linear algebra is used to propose
compressed vectors, after which their stored rational coefficients
are checked independently. The curves between evaluated densities
serve only to display the data.

\section{Finite-size results at \texorpdfstring{\(p=2/5\)}{p=2/5}}\label{sec:finite-size}

\subsection{A decrease followed by a certified increase}

Table~\ref{tab:exact} and Figure~\ref{fig:finite-size} give the
finite-size profile. The first seven probabilities were evaluated
with arbitrary-precision integers. The eighth was enclosed using a
fixed-point contraction with denominator \(2^{58}\).
The exact small-width values include
\begin{equation}\label{eq:small-exact}
 h_2(2/5)=\frac{184}{625},\qquad
 h_4(2/5)=\frac{1218364220691093248}{7450580596923828125}.
\end{equation}
Full integer numerators and denominators for all exactly evaluated
widths are supplied with the article.

\begin{table}[htbp]
\centering
\begin{tabular}{rll}
\toprule
\(L\)&\(h_L(2/5)\)&Arithmetic\\
\midrule
2&\(0.2944\) (exact)&integer\\
4&\(0.163526077577649125\ldots\)&integer\\
6&\(0.133941887841429917\ldots\)&integer\\
8&\(0.119109655957742335\ldots\)&integer\\
10&\(0.112799803761915568\ldots\)&integer\\
12&\(0.109734729315510052\ldots\)&integer\\
14&\(0.108905814034236649\ldots\)&integer\\
16&\((0.10924650244,\ 0.10924650247)\)&certified enclosure\\
\bottomrule
\end{tabular}
\caption{Initial digits of exact probabilities and an outward decimal
enclosure of the final interval. All widths refer to
\eqref{eq:geometry}.}
\label{tab:exact}
\end{table}

\begin{figure}[htbp]
\centering\includegraphics[width=\textwidth]{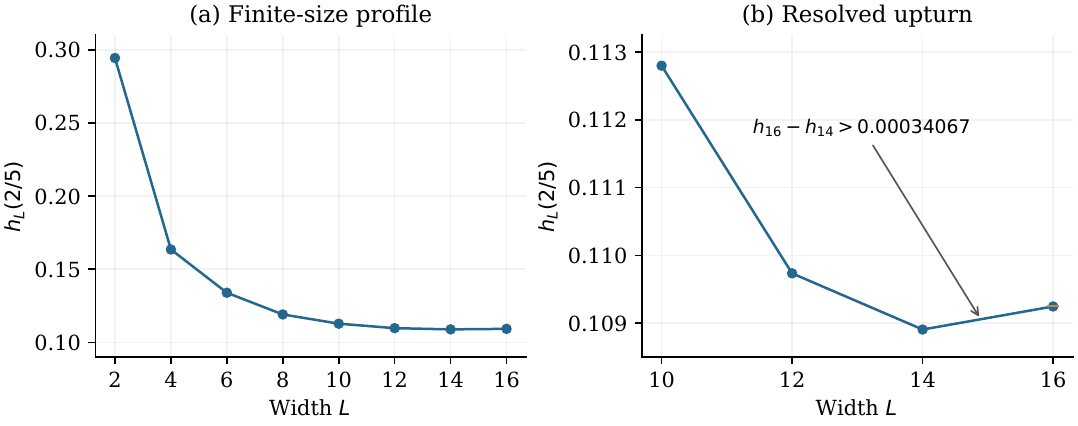}
\caption{Finite-size probability at \(p=2/5\). Panel (b) enlarges
the last four widths. The \(L=16\) enclosure has width less than
\(1.87\times10^{-11}\), below the plotted marker size. Lines
connect the measured widths.}
\label{fig:finite-size}
\end{figure}

The initial drop is substantial: from \(L=2\) to \(L=10\),
\(h_L\) decreases by more than a factor of two. The later decrease
is much slower, and reverses at the last tested step.
For exact comparison, use the common denominator
\[
 D_*=5\cdot2^{58}=1441151880758558720.
\]
The independently checked fixed-point outputs give
\begin{align}
 h_{14}(2/5)&\in
 \left[\frac{156949818717553904}{D_*},
       \frac{156949818724412336}{D_*}\right],\nonumber\\
 h_{16}(2/5)&\in
 \left[\frac{157440802463709882}{D_*},
       \frac{157440802490583738}{D_*}\right].\label{eq:intervals}
\end{align}
The first interval also contains the arbitrary-precision \(L=14\)
value. Subtracting the indicated rational endpoints proves
\begin{equation}\label{eq:increase}
 h_{16}(2/5)-h_{14}(2/5)>0.00034067.
\end{equation}
This increase is about \(0.31\%\) of the \(L=14\) probability,
and is far larger than either certified rounding bound.

\begin{proof}[Proof of Theorem~\ref{thm:nonmonotonicity}]
The exact fractions underlying Table~\ref{tab:exact} give the six
strict decreases by integer cross-multiplication. Their full numerators
and denominators are recorded in the supplement. The rational
enclosures \eqref{eq:intervals} give \eqref{eq:increase}, whose
right-hand side is \(34067/10^8\). Together these comparisons prove
the claimed nonmonotonicity and the unique minimum among the tested
widths. The tensor normalization and deterministic residual bounds
justifying the evaluations are proved in Appendices~\ref{sec:tensor}
and \ref{sec:certificates}.
\end{proof}

\subsection{What the observed upturn establishes}

Among the measured even widths \(2\) through \(16\), the minimum
occurs at \(14\). In particular, these finite-size probabilities
are not monotonically decreasing in width. Their definitions change
with \(L\): both the boundary and the connector endpoints move.
There is consequently no nesting of the measured events that would
force a one-sided finite-size trend.

The data do not resolve the large-width limit or the width at which
a high-probability confinement criterion might become useful.
The largest displayed width still gives approximately \(0.10925\),
well below \(0.8457\). We therefore use the upturn as a finite-size
observation, without extracting a localization length or fitting an
asymptotic power law from these eight points.

\section{Dependence on mirror density}\label{sec:density}

\subsection{An exact 105-point scan}

We evaluated every pair
\[
 L\in\{2,4,6,8,10\},\qquad p\in\{0,1/20,\ldots,19/20,1\}
\]
by integer contraction. Figure~\ref{fig:density} shows all 105 values;
Table~\ref{tab:density} lists representative densities. At each fixed
width, the measured values are strictly increasing along this grid.
This is an exact comparison of the sampled rational values, rather
than a proof of monotonicity between them. Individual changes of a
physical mirror can either create or remove odd boundary paths, as
the local experiments in Section~\ref{sec:local-experiments} demonstrate.

\begin{figure}[htbp]
\centering\includegraphics[width=\textwidth]{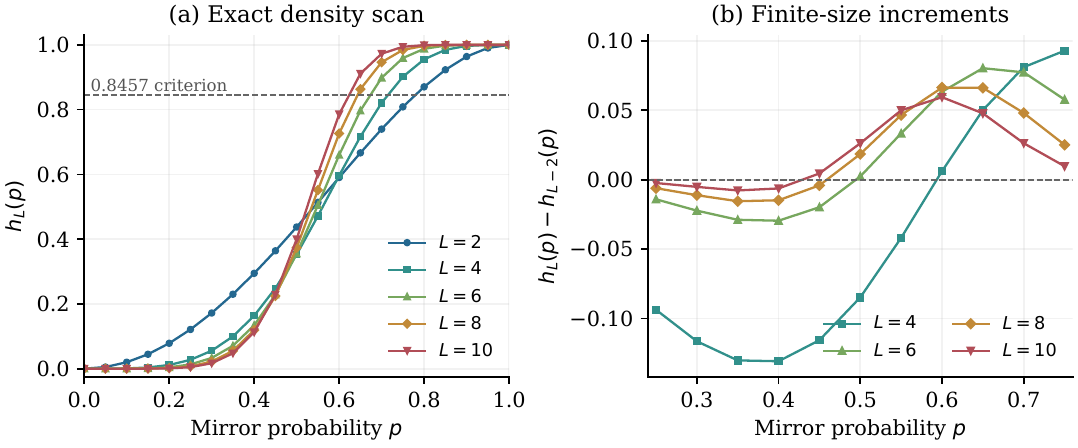}
\caption{Exact density scan. Panel (a) shows the boundary-parity
probabilities and the sufficient \(0.8457\) criterion. Panel (b)
shows adjacent measured-width differences. Each marker is an exact
rational evaluation; the connecting segments impose no interpolation
assumption. Crossings in panel (b) concern finite-width comparisons.}
\label{fig:density}
\end{figure}

\begin{table}[htbp]
\centering\small
\begin{tabular}{rllll}
\toprule
\(p\)&\(h_4(p)\)&\(h_6(p)\)&\(h_8(p)\)&\(h_{10}(p)\)\\
\midrule
0.25&0.027318147&0.013180467&0.007061335&0.004650719\\
0.40&0.163526078&0.133941888&0.119109656&0.112799804\\
0.45&0.248315813&0.228487985&0.224404006&0.228901192\\
0.50&0.352678254&0.355017988&0.373512164&0.399760138\\
0.65&0.716628892&0.797005653&0.863248673&0.911079801\\
0.80&0.955289155&0.987895208&0.996922873&0.999189847\\
\bottomrule
\end{tabular}
\caption{Selected density-scan values, rounded to nine decimal places. The underlying calculations and all comparisons use exact fractions.}
\label{tab:density}
\end{table}

The \(L=2\) scan has a separate calibration from all \(2^8=256\)
physical fields. Grouping successful fields by their number of
mirrors and expanding the resulting polynomial gives
\begin{equation}\label{eq:L2-polynomial}
 h_2(p)=2p^2-p^4.
\end{equation}
All 21 tensor values at this width agree with this independently
enumerated polynomial. The \(p=2/5\) values at every scanned width
also agree with the finite-size records in Table~\ref{tab:exact}.

\subsection{The finite-size trend depends on density}

At \(p=0.40\), the five widths used in this scan all lie on the
decreasing portion of the profile. At \(p=0.45\), however,
\[
 h_8(0.45)=0.224404005897\ldots,\qquad
 h_{10}(0.45)=0.228901192390\ldots,
\]
and the smallest scanned value occurs at \(L=8\).
At \(p=0.50\), the smallest scanned value occurs at \(L=4\),
with the values at \(6,8,10\) successively larger.
For example,
\begin{align*}
 h_{10}(0.40)-h_8(0.40)&\approx-0.006309852196,\\
 h_{10}(0.45)-h_8(0.45)&\approx\phantom{-}0.004497186494,\\
 h_{10}(0.50)-h_8(0.50)&\approx\phantom{-}0.026247974874.
\end{align*}
The signs of these differences were checked with exact fractions.

Each \(h_L(p)\) is a finite polynomial, so the first two signs imply
at least one crossing of \(h_{10}\) and \(h_8\) between \(0.40\)
and \(0.45\). The scan neither establishes uniqueness of this
crossing nor identifies it with a phase transition. Its practical
message is that a decreasing short-width profile and an increasing
longer-width profile can occur at the same fixed density.

\subsection{Reaching the sufficient criterion on the sampled grid}

The smallest sampled densities at which the finite probability
exceeds \(0.8457\) are as follows:
\begin{center}
\begin{tabular}{rccccc}
\toprule
Width \(L\)&2&4&6&8&10\\
First sampled \(p\)&0.80&0.75&0.70&0.65&0.65\\
\bottomrule
\end{tabular}
\end{center}
For example, \(h_8(0.65)=0.863248672529\ldots\) and
\(h_{10}(0.65)=0.911079800861\ldots\). These values exceed the
criterion by an exact rational comparison. They lie in a density
regime already covered by the known confinement result
\cite{LiManhattan}. The table records the resolution of this particular
finite-width scan; it does not give the smallest density at which
localization holds or at which a larger rectangle could certify it.

\section{Compression error and its certification}\label{sec:compression}

\subsection{A common benchmark for eight proposals}

The exact \(L=10\), \(p=2/5\) value provides a benchmark for
separating actual approximation error from a computable certificate.
We generated matrix-product-state proposals with nominal bond caps
\[
 \chi\in\{8,16,32,48,64,96,128,243\}.
\]
At every complete column, each proposal was converted to a full
coefficient vector with denominator \(Q=2^{40}\). All eight runs
used the same tensor, sweep order, density, and rounding convention.
The floating-point proposals were generated with one requested BLAS
thread; the error certificates subsequently used integer arithmetic.

Let \(\widehat h_\chi\) be the final stored coefficient and
\(e_\chi=|\widehat h_\chi-h_{10}(2/5)|\) its actual error.
For each run we evaluated two bounds from exactly the same residuals:
an entrywise bound \(E_{\ell^1}\) and a directed matching bound
\(E_{\rm match}\). Their construction is given in
Section~\ref{sec:methods}. They satisfy
\begin{equation}\label{eq:three-errors}
 e_\chi\le E_{\rm match}\le E_{\ell^1}.
\end{equation}
Here \(e_\chi\) is available because an exact benchmark was computed;
the two bounds can be checked without knowing that benchmark.

\begin{table}[htbp]
\centering
\begin{tabular}{rrrrr}
\toprule
\(\chi\)&Actual error \(e_\chi\)&\(E_{\rm match}\)&\(E_{\ell^1}\)&Gain\\
\midrule
8&\(6.34\times10^{-2}\)&\(3.49\times10^{0}\)&\(1.93\times10^{1}\)&5.52\\
16&\(3.16\times10^{-4}\)&\(2.51\times10^{-1}\)&\(2.07\times10^{0}\)&8.24\\
32&\(4.80\times10^{-5}\)&\(1.87\times10^{-2}\)&\(1.94\times10^{-1}\)&10.42\\
48&\(6.71\times10^{-6}\)&\(4.12\times10^{-3}\)&\(4.87\times10^{-2}\)&11.81\\
64&\(5.29\times10^{-7}\)&\(7.65\times10^{-4}\)&\(9.79\times10^{-3}\)&12.80\\
96&\(5.13\times10^{-9}\)&\(8.56\times10^{-5}\)&\(1.04\times10^{-3}\)&12.17\\
128&\(1.27\times10^{-11}\)&\(9.08\times10^{-6}\)&\(1.24\times10^{-4}\)&13.69\\
243&\(7.02\times10^{-14}\)&\(2.05\times10^{-9}\)&\(5.53\times10^{-8}\)&26.93\\
\bottomrule
\end{tabular}
\caption{Compression study. Error magnitudes are rounded displays of exact rational quantities. Gain is \(E_{\ell^1}/E_{\rm match}\). Bounds larger than one are shown before intersection with the probability range \([0,1]\).}
\label{tab:rank}
\end{table}

\begin{figure}[htbp]
\centering\includegraphics[width=\textwidth]{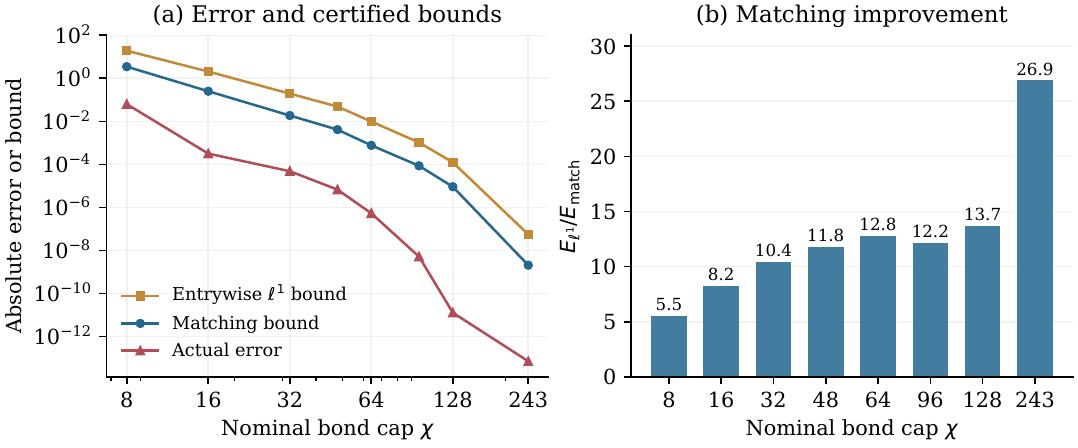}
\caption{Compression study at \(L=10\), \(p=2/5\), and
\(Q=2^{40}\). Panel (a) separates actual error from two deterministic
bounds on logarithmic axes. Panel (b) shows the improvement obtained
by applying matching constraints to the same residuals. The nominal
cap describes the floating-point proposal before coefficient rounding.}
\label{fig:compression}
\end{figure}

\subsection{Accuracy of the center versus strength of the bound}

Both the actual error and the two residual bounds decrease across
these eight runs, but their magnitudes differ substantially. With
\(\chi=16\), the center has error about \(3.16\times10^{-4}\),
while the matching bound is approximately \(0.251\).
At \(\chi=128\), the actual error is approximately
\(1.27\times10^{-11}\), yet the matching bound remains about
\(9.08\times10^{-6}\). Accurate central values alone would
therefore substantially overstate the precision that this residual
certificate establishes.

For \(\chi=64\), the certificate has the exact common denominator
\[
 D_{\rm c}=5^{10}2^{40}=10737418240000000000.
\]
Its center numerator is \(1211184346787109375\), and its summed
covering cost is \(8210572577426124\). Thus
\begin{equation}\label{eq:mps-interval}
 h_{10}(2/5)\in
 \left[\frac{1202973774209683251}{D_{\rm c}},
       \frac{1219394919364535499}{D_{\rm c}}\right].
\end{equation}
The matching radius is approximately \(0.0007646692\), compared
with \(0.0097877206\) for the entrywise radius. The improvement
factor is approximately \(12.8\), while the actual center error is
about \(5.29\times10^{-7}\).

At the largest cap, \(\chi=243\), the actual error is approximately
\(7.02\times10^{-14}\), with matching bound
\(2.05\times10^{-9}\). A length-ten, three-state coefficient
vector has a tensor-train representation with intermediate ranks at
most \(3^{\min(j,10-j)}\), whose maximum is \(243\)
\cite{Oseledets}. This run therefore provides a control at the full
representational rank. Floating-point operations and subsequent
coefficient rounding remain present.

\subsection{Proposal size and verification cost}

The nominal bond cap constrains the floating-point proposal. Rounding
its full coefficient vector does not imply that the resulting rational
vector has the same rank. Moreover, the certificate checker in this
experiment reconstructs all \(3^{10}=59049\) boundary coordinates
and their residuals. Its successful check does not establish an
algorithm with storage bounded by the nominal bond cap alone.

The proposal's largest number of core entries rises from \(1314\)
at \(\chi=8\) to \(36582\) at \(\chi=64\), and reaches
\(132858\) at \(\chi=243\). At large caps that count can exceed
the size of the full coefficient vector. We use this experiment to
measure approximation and certification, without claiming a general
memory or runtime advantage at every cap.

\section{Local-routing experiments}\label{sec:local-experiments}

\subsection{Exhaustive single-site changes}

We next examine a different, deliberately small geometry:
\(D=\{0,1,2,3\}^2\). One test marks the full seam through
\(y=3/2\), crossing all four vertical edges. A second marks only
the two edges at \(x=1,2\), so its connector endpoints lie inside
the box. Neither geometry is the \(8\times4\) rectangle \(D_4\).

For each test, we enumerated all \(2^{16}\) fields and compared
the two fillings of each physical site. There are
\(16\cdot2^{15}=524288\) distinct site/background toggle pairs.
To classify a pair, remove both passages at its site and follow its
two outgoing cavity routes until they meet the boundary or return
to a cavity input. Let \(k\in\{0,1,2\}\) be the number of
returns. Table~\ref{tab:cavity} reports the resulting counts and
whether the signed boundary map changes.

\begin{table}[htbp]
\centering
\begin{tabular}{rrrr}
\toprule
Returns \(k\)&Toggle pairs&Changed: full seam&Changed: interior connector\\
\midrule
0&442624&442624&442624\\
1&78848&0&43008\\
2&2816&0&0\\
\bottomrule
\end{tabular}
\caption{Complete \(4\times4\) enumeration. Each column counts
site/background pairs once; fractions formed from these counts use
uniform backgrounds, corresponding to \(p=1/2\) away from the
toggled site. The counts are not estimates at \(p=2/5\).}
\label{tab:cavity}
\end{table}

The two tests have the same unsigned cavity counts, because marking
an edge does not change its physical routing. Their signed responses
are different in the one-return class. For the full seam, all such
returns have even parity and all \(78848\) corresponding changes
are invisible at the boundary. For the interior connector, \(43008\)
of them are visible. The experiment isolates the role of connector
endpoints in a proposed local erasure rule. Appendix~\ref{sec:cavity}
proves the classification and the exact rank-one response underlying
these observations.

\subsection{A cooperative two-contact example}

In the full-seam box, put reference mirrors exactly at
\[
 (1,0),(2,0),(0,1),(1,1),(2,1),(3,1),
\]
and leave the two sites \((1,0),(1,1)\) available for toggling.
Forward and backward coordinate routing give the following four
outcomes, where an active contact means a change from its reference
mirror to a vacancy:
\begin{center}
\begin{tabular}{lrrrr}
\toprule
Active contacts&00&10&01&11\\
Odd boundary paths&0&0&0&2\\
\bottomrule
\end{tabular}
\end{center}
The first three signed boundary maps are identical. Thus the absence
of a response to either isolated change does not justify erasing
both variables simultaneously.

The reference geometry contains an even cycle with two distinct
boundary contacts. Only activation of both contacts exchanges the
two external tails. Conditional on the other fourteen bits, the
success probability is therefore
\begin{equation}\label{eq:cooperative-probability}
 1-(1-p)^2=\frac{16}{25}\quad\text{at }p=\frac25.
\end{equation}
The conditioned background has weight \(p^4(1-p)^{10}\); this
factor must be retained in any unconditional calculation.
Appendix~\ref{sec:contacts} gives the multi-contact formula that
reproduces this example. Its correction for the all-inactive state
ensures that an internal cycle has no artificial boundary penalty.

\section{Computational methods and verification}\label{sec:methods}

\subsection{Integer tensor contraction}

The local vacant and mirror routing tensors are denoted by
\(T_0\) and \(T_1^{x,y}\). Their three colors are two ordinary
colors and one signed color. Boundary vectors select even path parity,
while the sum over a closed cycle is exactly one, whether that cycle
has even or odd intersection with the connector. The explicit tensors,
boundary vectors, and planar sign argument appear in
Appendix~\ref{sec:tensor}; the rotation-number identity used there
is classical~\cite{Whitney}.

For \(p=a/b\), we contract the integer local tensor
\[
 (b-a)T_0+aT_1^{x,y}
\]
at every site and divide the final scalar by \(b^{2L^2}\).
There is one tensor choice per physical site, so this averaging
does not replace repeated passages by independent choices.
The site order is increasing \(x\), then increasing \(y\) within
each column. The sweep stores \(L\) horizontal colors and one
vertical carry, for a temporary dimension \(3^{L+1}\). Its
arithmetic operation count is \(O(L^2 3^{L+1})\), with growing
integer bit lengths in exact evaluations.

The density scan used NumPy object arrays containing arbitrary-size
Python integers. The \(L\le14\) reference values at \(p=2/5\)
used the same geometric formulation. The wider fixed-point sweep
used C++ integer arithmetic and explicit range checks before local
linear combinations. The different implementations use the same
specified physical model; their independent geometric and flattened
updates were compared at smaller widths.

\subsection{Deterministic fixed-point enclosures}

At \(p=2/5\), the normalized local tensor is
\((3T_0+2T_1)/5\). In fixed point, coefficients have denominator
\(Q=2^{58}\), and division by five is truncated towards zero after
each site. Every coordinate of the resulting local residual has
absolute value at most \(4/(5Q)\).

Group the free legs of an exact continuation by their physical
incoming and outgoing directions. A fixed environment gives a signed
partial permutation matrix; its Bernoulli average has absolute row
and column sums at most one. At every staircase cut the smaller
matrix dimension is \(3^{L/2}\). Summing the resulting matching
bound over the \(2L^2\) sites gives
\begin{equation}\label{eq:fixedpoint-main}
 |h_L(2/5)-\widehat h_L|
 \le\frac{8L^2 3^{L/2}}{5Q}.
\end{equation}
Appendices~\ref{sec:tensor} and \ref{sec:certificates} justify
these continuation constraints and the error summation.

The observed integer outputs \(Q\widehat h_L\) are
\(31389963744196624\) at \(L=14\) and
\(31488160495429362\) at \(L=16\). The error numerators in
\eqref{eq:fixedpoint-main}, relative to denominator \(5Q\), are
\(3429216\) and \(13436928\). These values yield
\eqref{eq:intervals}. The \(L=16\) temporary array contains
\(3^{17}=129140163\) signed 64-bit integers, or
\(1033121304\) bytes for that array alone.

\subsection{Residuals of compressed proposals}

Let \(u_x\) be the integer numerator vector stored after column
\(x\), with \(u_{-1}=Qe_{0\cdots0}\), and let \(C_x\)
be the exact integer column map at \(p=2/5\). Its residual
numerator is
\begin{equation}\label{eq:column-residual}
 E_x=5^L u_x-C_xu_{x-1}.
\end{equation}
For even \(L\), reshape this vector into a
\(3^{L/2}\times3^{L/2}\) matrix by putting odd horizontal heights
in one group and even heights in the other. Nonnegative integer
row and column potentials \(a_i^{(x)},b_j^{(x)}\) satisfying
\[
 a_i^{(x)}+b_j^{(x)}\ge |(E_x)_{ij}|
\]
give the computable bounds
\begin{equation}\label{eq:main-error-budgets}
 E_{\rm match}
 =\frac{\sum_x(\sum_i a_i^{(x)}+\sum_j b_j^{(x)})}{5^LQ},
 \qquad
 E_{\ell^1}=\frac{\sum_x\sum_{i,j}|(E_x)_{ij}|}{5^LQ}.
\end{equation}
The actual implementation chooses matching covers by an integer
assignment algorithm. The independent checker needs only their
feasibility, not their optimality.

At \(L=10\), every residual matrix is \(243\times243\).
For each proposal the checker reconstructs 200 geometric site
updates and checks \(20\cdot243^2=1180980\) covering inequalities,
as well as all coefficient dimensions, nonnegative potentials, and
final rational totals. Across the eight proposals this gives
\(9447840\) covering inequalities. No singular-value estimate or
floating-point optimizer output is used to validate a bound.
The common denominator in \eqref{eq:main-error-budgets} is the
denominator of a one-column residual; it is not the denominator
\(5^{2L^2}\) of a fully exact contraction.

\subsection{Reproducibility and independent checks}

The computational supplement contains all exact scan values, the
eight complete dyadic-vector and potential certificates, the original
\(L=14,16\) fixed-point logs, figure-generation scripts, and
coordinate-routing checks. The default verification checks the
saved numerical comparisons, all eight compression certificates,
and the small exact and physical-routing controls. Additional
commands recompute the 105-point scan and the larger integer sweeps.

The local verification includes 384 fixed fields compared between
coordinate routing and geometric tensor contraction, 541 subsets
of closed cycles checked for their signs, and 3072 directed cuts
checked for partial-permutation structure. Forward and backward
coordinate routers agree on 139264 complete small fields.
The single-site counts in Table~\ref{tab:cavity} are exhaustive,
while the multi-contact transfer is checked on 74896 active-field
cases and 2728 complete polynomial comparisons with up to five
contacts. Auxiliary signed-amplitude checks provide another routing
consistency test; circuit-nullity formulations of transition systems
are discussed in~\cite{Traldi}.

The software records exact fractions separately from display decimals.
Experiment scripts record their parameters and preserve the physical
weights of any conditioned background. The checks validate the
reported finite calculations; the general mathematical statements
supporting their certificates are proved in the appendices.

\section{Discussion}\label{sec:discussion}

Three observations emerge from the experiments. First, finite-size
boundary-parity probabilities can turn upward after a substantial
initial decrease. The increase at \(p=2/5\) is certified, and the
density scan shows changes in the measured location of the minimum.
Consequently, the smallest available widths do not support a uniform
decreasing extrapolation of this observable.

Second, geometric constraints on the continuation have a measurable
effect on numerical certification. Matching covers consistently
improve the entrywise bounds in the reported compression runs.
However, their gap from the actual error can remain large even when
the proposal center has many correct digits. Extending these
calculations to greater widths will require attention to the quality
and cost of the certificate as well as to the proposal itself.

Third, local simplification depends on which response is preserved.
The endpoint-sensitive one-return counts and the cooperative
two-contact example show why an isolated successful erasure cannot
automatically be applied to a set of interacting sites. Exact local
formulas can guide conditional computations provided that the
conditioning probabilities are retained.

The present \(p=2/5\) data do not approach the sufficient
\(0.8457\) threshold. No infinite-volume localization conclusion,
critical density, or scaling exponent is inferred from them. The
result of this study is a set of reproducible finite-size measurements
and quantified approximation errors, together with a verified method
for extending those measurements when larger computations become
feasible.

\section*{Acknowledgments}
The authors used OpenAI's GPT models to assist with mathematical
exploration, the development and review of proofs, drafting and
revising the manuscript, and checking bibliographic information.
The authors take full responsibility for all mathematical statements,
proofs, and the final content of this paper.

\clearpage
\appendix
\section{Routing and the finite-scale criterion}\label{app:criterion}
The appendices justify the exact observable and numerical certificates.
All domains, street phases, and connectors use the conventions of
Section~\ref{sec:model}. Internal states mean the two directed incoming
states at each physical site. We first record the finite routing
decomposition and the deterministic blocking argument.
\begin{lemma}[Finite routing decomposition]\label{lem:decomposition}
Every fixed environment in \(D\) decomposes its directed states into
disjoint boundary-input-to-boundary-output paths and internal directed
cycles. Each boundary input has one output, and every output is used
once.
\end{lemma}
\begin{proof}
Follow successors from a boundary input. A repeated state would put
that state on a cycle, whose predecessor on the cycle is already
occupied; tracing predecessors back to the boundary gives a
contradiction. The route therefore exits. Distinct such paths cannot
merge, by uniqueness of predecessors. Remove them. Every remaining
state has one predecessor and one successor among the remaining
states, and hence belongs to a cycle.
\end{proof}

\begin{lemma}[Modular intersection obstruction]\label{lem:blocking}
Fix one infinite physical environment. Let \(D_1,\ldots,D_k\) be finite
rectangular disks, with connectors \(\alpha_j\) contained in their
interiors. Suppose the connectors concatenate to a closed polygonal
walk \(C\). Let \(z\) lie outside every \(D_j\).
If, for some integer \(m\ge2\), every boundary path in \(D_j\) has
intersection with \(\alpha_j\) divisible by \(m\), while
\(\wind(C,z)\notin m\Z\), then every trajectory starting at \(z\)
is bounded. The witness disks may overlap.
\end{lemma}
\begin{proof}
If such a trajectory is unbounded, take a finite segment \(\gamma\)
from \(z\) to the unbounded exterior of a disk containing all witnesses.
Its portions in each \(D_j\) are complete boundary excursions.
By Lemma~\ref{lem:decomposition}, it cannot enter an internal cycle.
Consequently \(I(\alpha_j,\gamma)\in m\Z\) for every \(j\), with
multiplicity if there are several excursions. Additivity gives
\(I(C,\gamma)\in m\Z\).

On the other hand, the intersection of a closed polygonal walk with
a path from \(z\) to its unbounded exterior is, up to the choice of
sign convention, \(\wind(C,z)\). To see this, the winding number is
locally constant off \(C\), is zero in the unbounded component, and
changes by the signed crossing at each transverse intersection.
Summing these changes along \(\gamma\) gives the identity.
It contradicts the hypothesis. Small transverse perturbations resolve
coincidences at corners without changing either parity or winding.
\end{proof}

\begin{proof}[Proof of Theorem~\ref{thm:criterion}]
Partition the physical vertices into disjoint \(L\times L\) tiles,
indexed by \(\Z^2\). Use tile centers as the vertices of a coarse
square lattice. A coarse edge is open when every boundary path in
the rectangle consisting of its two endpoint tiles has even
intersection with their center-to-center connector.

The entire family of edge variables supported on a set of coarse
vertices is measurable with respect to the bits in those tiles.
Thus edge families with disjoint incident vertex sets are independent.
This proves the full \(1\)-independence required by
Balister--Johnston--Savery--Scott~\cite{BalisterEtAl}, rather than just
pairwise independence.

Translations by multiples of even \(L\) preserve the street phases.
The map
\[
 (x,y)\longmapsto(L-1-y,x)
\]
takes the horizontal witness to the vertical witness and reverses
every street arrow relative to our convention. Reversing all routes
restores the convention and preserves intersection parity.
It also preserves the product law on mirror bits.
Every coarse edge therefore has marginal probability \(h_L(p)\).

The surrounding-circuit conclusion of \cite{BalisterEtAl}
(Corollary~5.1 in the cited arXiv version) applies at
\(h_L(p)\ge0.8457\), including equality. Almost surely every bounded
set is enclosed by an open coarse cycle. For a fixed physical
starting state, enclose a sufficiently large neighborhood of it, so
that all two-tile witnesses along the cycle avoid the starting point.
The center connectors of a simple enclosing cycle have winding
\(\pm1\). Lemma~\ref{lem:blocking}, with \(m=2\), proves boundedness.
A countable intersection handles every starting state.
Finally, a bounded orbit of the bijective routing map is periodic.
\end{proof}

\section{Tensor normalization and directed continuations}\label{sec:tensor}

We represent the indicator, and hence the probability, of boundary
parity without solving a separate system for each environment.
All tensor contractions in this section are ordinary finite sums.
The third color implements the required signs through explicit local
coefficients.

\subsection{Local signs for closed curves}

Use colors \(0,1,\ff\), with standard basis
\(e_0,e_1,e_\ff\) of \(\R^3\). A color follows a routed strand.
At a vacant site put a minus sign if both crossing strands have color
\(\ff\). At a mirror put a minus sign for each \(\ff\)-colored
turn through the angular branch cut between east and south.
The latter turns occur precisely at sites with \(x\) odd and
\(y\) even.

The following elementary planar identity explains this convention.

\begin{lemma}[Loop sign]\label{lem:loop-sign}
Let a finite collection of oriented closed immersed curves in the
plane have only transverse double crossings, and let \(q\) be its
number of parameter-circle components. If \(C\) is the total number
of crossings and \(\mathrm{Rot}\) the sum of their tangent rotation
numbers (total turning divided by \(2\pi\)), then
\begin{equation}\label{eq:rotation-parity}
 C+\mathrm{Rot}\equiv q\pmod2.
\end{equation}
Consequently the local signs just specified have product
\((-1)^q\) on a collection of \(q\) closed \(\ff\)-colored
Manhattan trajectories.
\end{lemma}
\begin{proof}
Resolve one crossing by the orientation-preserving smoothing. This
changes the number of components by one modulo two and removes one
crossing. The two small connecting arcs can be chosen so that their
total tangent turning agrees with that of the original arcs; hence
the total rotation number is unchanged. After all crossings are
resolved, each component is an embedded oriented circle, with
rotation number \(+1\) or \(-1\). This proves
\eqref{eq:rotation-parity}. It is the parity form of the classical
rotation-number relation; see Whitney~\cite{Whitney}.

Round each mirror turn in a small neighborhood of its site, leaving
vacant passages as transverse crossings. The parity of crossings of
a generic tangent-angle branch cut equals the rotation number modulo
two. Our chosen cut lies strictly between east and south, and is
crossed exactly by the stated mirror turns. The vacant-site signs
supply \((-1)^C\), and the turn signs supply
\((-1)^{\mathrm{Rot}}\). Their product is \((-1)^q\).
\end{proof}

Distinct routed curves never share a directed edge. The rounded
curves therefore satisfy the hypotheses after arbitrarily small
local perturbations. If open curves are also present in a planar
disk and have both endpoints on its boundary, a closed curve has an
even total number of transverse intersections with each open curve.
Indeed, its mod-two winding number is zero at both endpoints, and
each intersection changes that number. We will use this observation
when cut legs are left free.

\subsection{The tensor and the boundary normalization}

Write the four geometric leg colors at a site as \((w,s,e,n)\),
in west, south, east, north order. Set
\(\eta_{x,y}=\ind_{\{x\text{ odd},\ y\text{ even}\}}\).
The vacant and mirror tensors are
\begin{align}
 T_0(w,s,e,n)
 &=(-1)^{\ind_{\{w=\ff\}}\ind_{\{s=\ff\}}}
   \delta_{w,e}\delta_{s,n},\label{eq:T0}\\
 T_1^{x,y}(w,s,e,n)
 &=\begin{cases}
   \delta_{w,n}\delta_{s,e},&x+y\text{ even},\\
   (-1)^{\eta_{x,y}(\ind_{\{w=\ff\}}+\ind_{\{e=\ff\}})}
      \delta_{w,s}\delta_{e,n},&x+y\text{ odd}.
  \end{cases}\label{eq:T1}
\end{align}
On each marked edge insert
\[
 Z=\diag(1,-1,-1),
\]
and place the vector \(v=(e_0+e_1)/\sqrt2\) at every exterior leg.
Contract colors on each internal edge with the usual Euclidean pairing.

\begin{theorem}[Exact normalized representation]\label{thm:tensor}
For every fixed physical environment, the contraction of
\eqref{eq:T0}--\eqref{eq:T1}, with the specified markers and boundary
vectors, equals \(\ind_{H(D,\alpha)}\). Replacing the tensor at
each site by
\begin{equation}\label{eq:averaged-tensor}
 T_p^{x,y}=(1-p)T_0+pT_1^{x,y}
\end{equation}
therefore gives \(\Pp(H(D,\alpha))\).
Each internal cycle has total weight exactly one, independently of
its intersection parity.
\end{theorem}
\begin{proof}
For a fixed environment, colors are constant along routed strands
before marker factors are applied. Boundary vectors exclude color
\(\ff\) on every boundary path. A boundary path with \(k\) marked
traversals contributes
\[
 v^\top Z^k v=\frac{1+(-1)^k}{2}.
\]
Thus each path contributes its own even-parity indicator, with no
normalization depending on its length.

An internal cycle with \(k\) marked traversals has ordinary-color
contribution \(1+(-1)^k\). By Lemma~\ref{lem:loop-sign}, choosing
color \(\ff\) on that cycle contributes \(-(-1)^k\).
More precisely, for every subset of cycles colored \(\ff\), the
joint local sign is the product of one minus sign per chosen cycle;
crossings between chosen cycles introduce no further dependence.
Summing all choices therefore gives one per internal cycle.
This proves the fixed-environment assertion. Multilinearity and
independence of the physical site bits give
\eqref{eq:averaged-tensor}. In particular, the two passages at a
site are averaged together through the same tensor choice.
\end{proof}

This is a finite three-color implementation of the loop cancellation
mechanism appearing in~\cite{MartinsNienhuisRietman,JacobsenReadSaleur,NahumEtAl}.
The normalization in Theorem~\ref{thm:tensor} specifies the particular
boundary event used here.

For computation, change the ordinary-color basis orthogonally by
\[
 e_0\longmapsto\frac{e_0+e_1}{\sqrt2},\qquad
 e_1\longmapsto\frac{e_0-e_1}{\sqrt2},\qquad
 e_\ff\longmapsto e_\ff.
\]
The pairing tensors \eqref{eq:T0}--\eqref{eq:T1} are unchanged: each
ordinary-color pairing is invariant under this simultaneous orthogonal
change, and the \(\ff\) sector is fixed. Boundary vectors become
\(e_0\), and the edge marker becomes the signed permutation
\begin{equation}\label{eq:X}
 Xe_0=e_1,\qquad Xe_1=e_0,\qquad Xe_\ff=-e_\ff.
\end{equation}
We use this basis from now on. It eliminates all square roots from
the contraction.

\subsection{An explicit column sweep}
For \(p=2/5\), abbreviate \(Z_L=Z_L(2,5)\) in \eqref{eq:intro-integer}.

Sweep \(D_L\) through columns \(x=0,\ldots,2L-1\), processing
\(y=0,\ldots,L-1\) within each column. Keep one horizontal color at
each height and one vertical carry color. The temporary vector has
\(3^{L+1}\) coordinates. At a column entrance, the carry is \(0\).
Before processing \((x,y)\), apply \(X\) to the carry exactly when
\[
 c\le x<3c,\qquad y=c.
\]
This inserts the marker on the vertical edge that the sweep is about
to cross, regardless of the physical orientation of that edge.

At \(p=2/5\), fix all unaffected coordinates and write \(V[a,d]\)
for the two-index slice consisting of the horizontal input color and
the carry. For \(b,t\in\{0,1,\ff\}\), the integer tensor
\(3T_0+2T_1^{x,y}\) gives
\begin{equation}\label{eq:local-sweep}
 V'[b,t]=3(-1)^{\ind_{\{b=\ff\}}\ind_{\{t=\ff\}}}V[b,t]
 +\begin{cases}
   2V[t,b],&x+y\text{ even},\\[2pt]
   2\ind_{\{b=t\}}s_b\displaystyle\sum_a s_a V[a,a],
       &x+y\text{ odd},
  \end{cases}
\end{equation}
where \(s_a=(-1)^{\eta_{x,y}\ind_{\{a=\ff\}}}\).
At the top of each column project the carry onto \(0\), and reset
the next column's bottom carry to \(0\). Initially all horizontal
colors are \(0\); finally select their all-zero coordinate.

\begin{corollary}[Integer evaluation]\label{cor:integer}
Starting with the all-zero coordinate equal to one, the sweep
\eqref{eq:local-sweep} produces an integer \(Z_L\) satisfying
\eqref{eq:intro-integer}. It uses \(O(L^2 3^{L+1})\) arithmetic
operations and \(O(3^{L+1})\) scalar storage.
\end{corollary}
\begin{proof}
Each update is the contraction of the next local tensor, and each
marker and boundary leg is inserted exactly once. There are
\(2L^2\) site tensors, each with denominator five in
\eqref{eq:averaged-tensor}. Theorem~\ref{thm:tensor} gives the
claimed scalar. A fixed-size local slice needs a bounded number of
operations, which gives the stated counts.
\end{proof}

The operation count is in arithmetic operations; integer bit lengths
increase with the number of sites. Both the exact sweep and its
uncompressed fixed-point version retain exponential storage in \(L\).

\subsection{The exact continuation at a cut}

Consider a rectangle, or a simply connected staircase region left by
the column sweep. Some boundary legs are free cut legs; all other
boundary legs carry \(e_0\). Group the free legs according to whether
their physical arrows point into or out of the remaining region.
In the basis \eqref{eq:X}, its continuation tensor becomes a matrix
\(W\), with one index for each color assignment to the incoming
group and one for each assignment to the outgoing group. Either
choice of which group indexes rows is permitted, provided it is used
consistently.

A \emph{signed partial permutation matrix} has entries in
\(\{0,1,-1\}\) and at most one nonzero entry in each row and column.
It can be rectangular or identically zero.

\begin{theorem}[Directed continuation]\label{thm:continuation}
For every fixed physical environment in the remaining region,
\(W\) is a signed partial permutation matrix. For independent real
Bernoulli parameters in \([0,1]\), their averaged continuation
satisfies
\begin{equation}\label{eq:substochastic}
 \sum_j|W_{ij}|\le1\quad\hbox{for every }i,
 \qquad
 \sum_i|W_{ij}|\le1\quad\hbox{for every }j.
\end{equation}
In particular, \(\|W\|_{2\to2}\le1\).
\end{theorem}
\begin{proof}
Decompose a fixed routing into open strands and internal cycles.
An open strand may connect two cut legs, a cut leg to a fixed exterior
leg, or two fixed exterior legs. Every strand connects a physical
input to a physical output. Its color transport, including markers,
is a signed permutation of the three colors.

An assignment to all incoming cut legs therefore determines at most
one compatible assignment to the outgoing cut legs. Strands ending
at fixed exterior legs may disallow this assignment. The same
argument with paths followed backwards proves uniqueness in each
column. If compatible, the open-strand contribution has magnitude one.

It remains to sum internal cycles in the presence of possibly
\(\ff\)-colored open strands. Every closed cycle has even total
intersection with each such open strand, because both endpoints of
the open strand are on the boundary of the disk. Its vacant-crossing
sign with that strand is consequently \(+1\). Lemma~\ref{lem:loop-sign}
and the proof of Theorem~\ref{thm:tensor} then give weight one for
every internal cycle, also in this partially contracted region.
Thus the surviving matrix entries are exactly \(\pm1\).

Average over the actual nonnegative product weights of the remaining
site variables, which sum to one. The triangle inequality gives
\eqref{eq:substochastic}. Finally, the usual row-and-column norm
bound, or Cauchy--Schwarz applied to \((|W_{ij}|)\), gives
\(\|W\|_{2\to2}^2\le\|W\|_1\|W\|_\infty\le1\).
\end{proof}

The grouping by physical direction is essential to the statement.
Grouping neighboring tensor indices without regard to their arrows
need not produce a partial permutation.

\section{Residual error bounds}\label{sec:certificates}

\subsection{Matching and rational covering potentials}

For a finite real or complex matrix \(A\), define
\begin{equation}\label{eq:matching}
 \cM(A)=\max_M\sum_{(i,j)\in M}|A_{ij}|,
\end{equation}
where the maximum is over partial matchings of row and column indices:
no row or column is used more than once. This is the usual weighted
bipartite matching problem~\cite{Birkhoff,Kuhn}.

\begin{lemma}[Residual bound]\label{lem:matching}
If \(W\) satisfies \eqref{eq:substochastic}, then
\begin{equation}\label{eq:matching-bound}
 \left|\sum_{i,j}W_{ij}A_{ij}\right|\le\cM(A).
\end{equation}
If nonnegative numbers \(a_i,b_j\) satisfy
\begin{equation}\label{eq:cover}
 a_i+b_j\ge |A_{ij}|\qquad\hbox{for all }i,j,
\end{equation}
then
\begin{equation}\label{eq:cover-bound}
 \cM(A)\le\sum_i a_i+\sum_j b_j.
\end{equation}
For a matrix with \(r\) rows and \(s\) columns, writing
\(k=\min(r,s)\), one also has
\begin{equation}\label{eq:simple-bounds}
 \cM(A)\le
 \min\left\{\sum_{i,j}|A_{ij}|,\ \sqrt{k}\|A\|_F,
                      \ k\max_{i,j}|A_{ij}|\right\}.
\end{equation}
\end{lemma}
\begin{proof}
Put \(Y_{ij}=|W_{ij}|\). The polytope of nonnegative rectangular
matrices with row and column sums at most one has only partial
permutation matrices as extreme points. For completeness, in the
support of the strictly fractional entries of a putative fractional
extreme point, a cycle admits an alternating positive and negative
perturbation that preserves every row and column sum. If there is no
cycle, a nontrivial tree component has two leaves. Each leaf has
strictly positive unused capacity: its only nonzero incident entry
is fractional, since an incident entry equal to one would exclude all
others. The path between the leaves admits a small alternating
perturbation, preserving internal sums and respecting the endpoint
capacities. In both cases the point is not extreme. This proves the
assertion, including an isolated fractional edge as a one-edge path.

Maximizing the linear functional
\(\sum Y_{ij}|A_{ij}|\) over this polytope now gives
\eqref{eq:matching-bound}. Summing \eqref{eq:cover} along a matching,
using each nonnegative potential at most once, proves
\eqref{eq:cover-bound}. The first and third bounds in
\eqref{eq:simple-bounds} follow immediately, and the second follows
from Cauchy--Schwarz on the at most \(k\) matched entries.
\end{proof}

The matching value is the exact support bound for the row-and-column
relaxation \eqref{eq:substochastic}. We do not assert that every
matrix in that relaxation can be realized by a physical continuation.
The nonnegativity of both sets of potentials in \eqref{eq:cover}
is needed for a certificate on partial matchings.

\begin{theorem}[Certified approximate contraction]\label{thm:certificate}
Let \(A_1,\ldots,A_N\) be the exact normalized transfer maps of a
planar column or staircase sweep, including its exact marker and
boundary operations. Suppose the initial vector \(v_0\) is exact,
and arbitrary approximate vectors satisfy
\[
 \widehat v_0=v_0,\qquad
 r_i=\widehat v_i-A_i\widehat v_{i-1}.
\]
At the \(i\)-th cut, reshape \(r_i\) into the matrix \(R_i\) by
physical direction. Let \(z\) be the exact final contraction and
\(\widehat z\) the same final boundary functional applied to
\(\widehat v_N\). Then
\begin{equation}\label{eq:certificate}
 |z-\widehat z|\le\sum_{i=1}^N\cM(R_i).
\end{equation}
In particular, nonnegative rational potentials
\(a^{(i)},b^{(i)}\) covering each exact rational residual as in
\eqref{eq:cover} certify
\begin{equation}\label{eq:certificate-dual}
 |z-\widehat z|
 \le\sum_i\left(\sum_j a_j^{(i)}+\sum_k b_k^{(i)}\right).
\end{equation}
\end{theorem}
\begin{proof}
Let \(w_i\) denote the exact, fully contracted suffix beyond cut
\(i\). The identity \(w_{i-1}^\top=w_i^\top A_i\) gives
\[
 \widehat z-z=\sum_{i=1}^N w_i^\top r_i
\]
by telescoping. Each matricized suffix satisfies
\eqref{eq:substochastic} by Theorem~\ref{thm:continuation}.
Apply Lemma~\ref{lem:matching} to every summand.
\end{proof}

This certificate places no assumption on how the approximate vectors
were obtained. They may be produced by tensor-train rounding
\cite{Oseledets}, a variational method, or another proposal algorithm.
When stored rationally, their exact residuals and covering
inequalities can be checked without trusting an optimization routine.
The signed local transfer maps themselves need not be contractions
in the coefficient norm used by that routine.

\subsection{A uniform fixed-point certificate}

At a complete column cut there are \(L/2\) physically incoming and
\(L/2\) outgoing horizontal legs. The residual matrix is therefore
\(3^{L/2}\times3^{L/2}\). Between sites there is one additional
vertical carry leg, so the two group sizes are \(L/2\) and
\(L/2+1\). At every such cut the smaller matrix dimension is
\(3^{L/2}\).

\begin{corollary}[Fixed-point rounding]\label{cor:fixed-point}
Run the normalized sweep at \(p=2/5\) with every coefficient an
integer multiple of \(Q^{-1}\), where \(Q\) is a positive integer.
After each site update, divide its integer numerator by five and
truncate towards zero; carry markers and boundary projections are
exact. If \(\widehat h_L\) is the resulting final coefficient, then
\begin{equation}\label{eq:fixed-point}
 |h_L(2/5)-\widehat h_L|
 \le \frac{8L^2 3^{L/2}}{5Q}.
\end{equation}
\end{corollary}
\begin{proof}
Each divided integer has remainder of absolute value at most four,
so every coordinate of a local residual is at most \(4/(5Q)\) in
absolute value. Exact marker permutations and coordinate projections
preserve this bound. By \eqref{eq:simple-bounds}, its matching value
is at most \(4\cdot3^{L/2}/(5Q)\). There are \(2L^2\) site
updates. Apply Theorem~\ref{thm:certificate}.
\end{proof}

Arithmetic overflow must separately be excluded in any implementation
using bounded integers. The supplement's fixed-point implementation
checks the coefficient range before local integer combinations and
uses parameters for which those combinations fit in its integer type.

\section{Local cavity response}\label{sec:cavity}

The preceding certificates use a global continuation constraint.
We now describe an exact local simplification, keeping the original
physical Bernoulli variables throughout.

\subsection{The three cavity configurations}

Let \(v\) be a site in a finite rectangular disk \(B\). Expose the
bits in \(B\setminus\{v\}\), but leave \(\omega_v\) unread.
Remove both local transitions at \(v\). There are two cavity inputs
and two cavity outputs. Starting at each output, follow the exposed
routing until it reaches a cavity input or an exterior output of
\(B\). Such a route cannot enter an internal cycle, by uniqueness
of predecessors. Let \(k\in\{0,1,2\}\) be the number of returns
to cavity inputs. A return is a path in the punctured region; its
possible closing passage at \(v\) is specified by a filling.

Let \(S_0,S_1\) be the signed boundary scattering matrices of \(B\)
with the cavity filled by vacancy and mirror, respectively.

\begin{proposition}[Cavity classification]\label{prop:cavity}
The following alternatives exhaust all possibilities.
\begin{enumerate}
 \item If \(k=0\), two distinct boundary inputs reach the cavity,
 and two cavity outputs reach distinct boundary outputs. The two
 fillings exchange the two outgoing tails.
 \item If \(k=1\), one filling leaves one boundary path and a
 separate internal cycle; the other inserts that cycle's return
 route into the boundary path. If the closed return has sign
 \(\ell\), the affected boundary signs differ by \(\ell\).
 \item If \(k=2\), the cavity belongs entirely to internal cycles,
 and \(S_0=S_1\).
\end{enumerate}
In the second case \(S_0=S_1\) if and only if \(\ell=1\).
\end{proposition}
\begin{proof}
The exposed routing pairs its remaining input and output terminals.
Two different traced routes cannot meet a directed state. If no
cavity output returns, both cavity inputs must have predecessors on
two distinct boundary paths; bijectivity also makes their two
exterior outputs distinct. Changing the filling exchanges the
connections between these two input and output routes.

If one output returns to an input, one filling connects that input
back to the same output and hence closes a cycle. Its other passage
connects the remaining boundary input and output routes. The other
filling joins these pieces into one boundary route, multiplying its
sign by the closed return's sign. If both outputs return, every
cavity terminal is already connected within this four-terminal
system. No boundary path reaches it, so changing its internal cycle
decomposition leaves the boundary map unchanged.
\end{proof}

\begin{corollary}[Four arms away from the endpoints]\label{cor:four-arms}
Suppose neither endpoint of \(\alpha\) lies in \(B\). If
\(S_0\ne S_1\), the punctured routing has \(k=0\): two incoming
and two outgoing routes join the cavity to \(\partial B\), disjoint
as directed states. Thus \(k\ge1\) certifies \(S_0=S_1\) without
reading the bit at \(v\).
\end{corollary}
\begin{proof}
Every closed curve contained in \(B\) has even intersection with
\(\alpha\). Indeed, both endpoints of the connector lie outside
the disk and hence have the same, zero, winding number with respect
to that curve. In the \(k=1\) case the closed return therefore has
sign \(+1\). Apply Proposition~\ref{prop:cavity}.
\end{proof}

These are four \emph{routed} arms, not independent percolation arms.
Two of them may pass through the same physical site on its different
incoming states and consequently use the same Bernoulli variable.

\begin{proposition}[Exact erasure]\label{prop:erasure}
Suppose an exposed configuration in \(B\setminus\{v\}\) certifies
\(S_0=S_1\). For every fixed exterior environment in a larger finite
domain \(D\), the boundary-parity indicator is independent of
\(\omega_v\). The original bit can therefore be summed out exactly,
also with a complex site parameter \(z\):
\begin{equation}\label{eq:erasure}
 (1-z)\ind_{H(D,\alpha)}\big|_{\omega_v=0}
 +z\ind_{H(D,\alpha)}\big|_{\omega_v=1}
 =\ind_{H(D,\alpha)}\big|_{\omega_v=0}.
\end{equation}
\end{proposition}
\begin{proof}
Replace each complete excursion through \(B\) by its signed boundary
connection. Equality of the two signed scattering matrices makes
this replacement identical for both fillings. Composing with the
fixed exterior routing gives the same global boundary map, including
paths making several excursions into \(B\). Excursions that close
into internal cycles impose no boundary condition. The two indicators
are equal, and \eqref{eq:erasure} follows from \((1-z)+z=1\).
\end{proof}

Erasure removes the weight of the unread bit only. Every bit exposed
to obtain the certificate retains its own original probability factor.
In particular, selecting a favorable surrounding configuration does
not make that configuration a probability-one event.

\subsection{Rank-one response and its exact strength}

\begin{theorem}[Single-site reflection]\label{thm:reflection}
For arbitrary fixed edge signs and any cavity,
\begin{equation}\label{eq:rank-one}
 \rank(S_1-S_0)\le1.
\end{equation}
If the difference is nonzero, there is a real unit vector \(u\),
supported on at most two output coordinates, such that
\[
 S_1=(I-2uu^\top)S_0.
\]
In this nontrivial case, for a Bernoulli bit of parameter \(p\),
its averaged map satisfies
\begin{equation}\label{eq:reflection-average}
 \overline S=(I-2puu^\top)S_0,\qquad
 \overline S^\top\overline S
 =I-4p(1-p)vv^\top,\qquad v=S_0^\top u.
\end{equation}
It has one affected singular value \(|1-2p|\), with all other
singular values equal to one.
\end{theorem}
\begin{proof}
In the \(k=0\) case, restrict to the two affected input and output
coordinates. Prefix and suffix signs give diagonal sign matrices
\(D_a,D_b\) such that, after relabeling, the two maps are
\(D_b I D_a\) and \(D_b P D_a\), where
\(P=\left(\begin{smallmatrix}0&1\\1&0\end{smallmatrix}\right)\).
Their relative map is \(D_bPD_b\), a symmetric orthogonal
reflection with eigenvalues \(+1,-1\).
If \(k=1\) and the closed return is odd, the relative map flips
one output sign. In every other case it is the identity, by
Proposition~\ref{prop:cavity}. This proves the rank and reflection
claims. Averaging the reflection and expanding its Gram matrix gives
\eqref{eq:reflection-average} and the singular values.
\end{proof}

At \(p=2/5\), the affected singular value is \(1/5\). This is a
statement about one mode of an averaged amplitude map. It is not a
probability that a trajectory is removed, and the unaffected modes
remain undamped.

\section{One internal cycle with several independent contacts}\label{sec:contacts}

Single-site erasure does not imply that a collection of sites may be
erased simultaneously. Two changes can create a response even when
each change alone is invisible in a reference environment. For one
reference cycle, this interaction has an exact finite-dimensional
description.

\subsection{Hypotheses and the transfer formula}

Condition on all bits except those at distinct candidate sites
\(v_1,\ldots,v_m\), where \(m\ge1\). Choose their reference
values without reading these free bits, and consider the resulting
reference routing in a finite domain. Assume the following:
\begin{enumerate}
 \item One internal cycle \(C\) visits each candidate in exactly
 one of its two directed states, in the cyclic order
 \(v_1,\ldots,v_m\).
 \item The other state at \(v_i\) belongs to a boundary path
 \(\gamma_i\). These \(m\) paths are distinct, and each contains
 no other candidate state.
 \item All routing outside the candidate transitions is fixed.
\end{enumerate}
Call a site inactive when it takes its reference value and active
when its two successors are exchanged. Let \(\rho_i\) be its
inactive probability and \(s_i=1-\rho_i\) its active probability.
For the original parameter \(p\),
\[
 \rho_i=p\quad\hbox{at a reference mirror},\qquad
 \rho_i=1-p\quad\hbox{at a reference vacancy}.
\]
The candidate bits are conditionally independent. The probability of
the exposed background is separate and is not included in the
conditional formulas below.

Let \(\ell\in\{\pm1\}\) be the sign of \(C\). A sign gauge on
the cycle makes every inter-contact transport positive except the
closing transport from \(m\) to \(1\), whose sign is \(\ell\).
In this gauge, let \(a_i\) be the external prefix sign and \(b_i\)
the external suffix sign at contact \(i\). Thus the reference path
\(\gamma_i\) has sign \(a_ib_i\). An active contact sends external
input amplitude one into the cycle with sign \(a_i\), and requires
cycle input sign \(b_i\) for its external output to equal one.

Index the standard basis of \(\R^2\) by \(+1,-1\), writing
\(e_+,e_-\), and set
\begin{align}
 J_+&=I,&
 J_-&=\begin{pmatrix}0&1\\1&0\end{pmatrix},\nonumber\\
 M_i&=\rho_i\ind_{\{a_i=b_i\}}I+s_i e_{a_i}e_{b_i}^\top,
 & Q&=J_\ell M_m\cdots M_1.\label{eq:contact-transfer}
\end{align}

\begin{theorem}[Exact one-cycle contact formula]\label{thm:contacts}
Under these hypotheses, the conditional probability that all affected
boundary paths have even parity is
\begin{align}
 \cR
 &=\tr Q-\ell\prod_{i=1}^m
                  \rho_i\ind_{\{a_i=b_i\}}\label{eq:trace-correction}\\
 &=\tr Q-\det Q=1-\det(I-Q).\label{eq:trace-det}
\end{align}
The full conditional boundary-parity probability is \(\cR\)
multiplied by the fixed success indicator of the unaffected boundary
paths. The identities remain valid as polynomial identities for
complex \(\rho_i\).
\end{theorem}
\begin{proof}
First fix the active subset \(A\). An inactive contact retains its
reference external path, requiring \(a_i=b_i\), while transporting
the cycle sign unchanged. An active contact forces the cycle sign
just before it to be \(b_i\) and the sign just after it to be
\(a_i\). These are exactly the two local matrices whose weighted
sum is \(M_i\).

If \(A\ne\varnothing\), one active contact fixes the sign on the
cycle, after which every other sign is determined. The trace with
the closing matrix \(J_\ell\) is therefore either zero or one.
It is one exactly when every newly connected boundary path has
positive sign. Equivalently, an active input at \(i\) follows the
reference cycle to the next active contact \(j\), exiting with sign
\(a_ib_j\) times the intervening cycle transport; the trace imposes
precisely that sign condition.

For \(A=\varnothing\), the reference internal cycle is unrestricted.
If some \(a_i\ne b_i\), both the correct indicator and the trace
vanish. If all are equal, the correct indicator is one, whereas the
trace of \(J_\ell\) is two for \(\ell=1\) and zero for
\(\ell=-1\). Subtracting \(\ell\) times the all-inactive weight
corrects exactly this discrepancy. Summing independent candidate
weights proves \eqref{eq:trace-correction}.

If \(a_i=b_i\), the eigenvalues of \(M_i\) are \(1,\rho_i\);
otherwise \(M_i\) has rank at most one and determinant zero.
Thus
\[
 \det Q=\det J_\ell\prod_i\det M_i
 =\ell\prod_i\rho_i\ind_{\{a_i=b_i\}}.
\]
For a \(2\times2\) matrix,
\(\det(I-Q)=1-\tr Q+\det Q\), proving
\eqref{eq:trace-det}. All operations are polynomial, so the same
identities hold at complex parameters.
\end{proof}

The correction in \eqref{eq:trace-correction} is the same boundary-only
normalization principle as in Theorem~\ref{thm:tensor}: an internal odd cycle must not cause the
boundary event to fail.

\begin{corollary}[Even reference boundary paths]\label{cor:contact-cases}
Suppose \(a_i=b_i\) for all \(i\). If \(\ell=1\), put
\[
 P_+=\prod_{i:a_i=1}\rho_i,\qquad
 P_-=\prod_{i:a_i=-1}\rho_i,
\]
with empty products equal to one. Then
\begin{equation}\label{eq:two-phase}
 \cR=P_++P_--P_+P_-,\qquad
 1-\cR=(1-P_+)(1-P_-).
\end{equation}
If \(\ell=-1\), then
\begin{equation}\label{eq:odd-loop-contacts}
 \cR=\prod_i\rho_i.
\end{equation}
If at least one reference boundary path is odd, then
\(\det Q=0\), \(\rank Q\le1\), and \(\cR=\tr Q\).
\end{corollary}
\begin{proof}
When \(a_i=b_i\), all \(M_i\) are diagonal. For \(\ell=1\),
their product has diagonal entries \(P_-,P_+\), giving
\eqref{eq:two-phase}. For \(\ell=-1\), multiplication by
\(J_-\) makes the trace zero and the determinant
\(-\prod_i\rho_i\). If some \(a_i\ne b_i\), the corresponding
factor has rank at most one and determinant zero.
\end{proof}

\subsection{Two contacts and cooperative response}

\begin{corollary}[An effective two-contact exchange]\label{cor:exchange}
Suppose \(m=2\) and \(\ell=1\). If no contact or only one contact
is active, the signed boundary scattering equals the reference
scattering. If both are active, the two external tails are exchanged.
Thus the exact conditional scattering law is an unchanged routing
with probability \(1-s_1s_2\), and the signed exchanged routing with
probability \(s_1s_2\).
\end{corollary}
\begin{proof}
With one active contact, the affected path traverses the whole
reference cycle before returning to its own suffix. The inserted
sign is \(\ell=1\), so its boundary connection and sign are
unchanged. With both contacts active, each input exits at the other
contact. Independence gives the product \(s_1s_2\).
\end{proof}

The result is a four-terminal exchange law. It does not identify a
new square-lattice model without a further geometric construction.
For several such gadgets, independence of the effective variables
requires disjoint free physical bits and a surrounding geometry
determined without those bits.

\end{document}